\documentclass[acmsmall]{acmart}

\usepackage{microtype}
\usepackage[utf8]{inputenc}
\usepackage[T1]{fontenc}
\usepackage{amsmath}
\usepackage{amsfonts}
\usepackage{bbm}
\usepackage{algorithmic}
\usepackage[linesnumbered,ruled,vlined]{algorithm2e}
\usepackage{graphicx}
\usepackage{textcomp}
\usepackage{xcolor,color,colortbl}
\usepackage{xspace}
\usepackage{cleveref}
\usepackage{hyperref}
\usepackage{wrapfig}
\usepackage{multirow}
\usepackage{xcolor}  
\usepackage{tikz}
\usepackage{enumitem}
\usepackage{float}

\AtBeginDocument{%
  }

\setcopyright{cc}
\setcctype{by}
\acmDOI{10.1145/3832252}
\acmYear{2026}
\acmJournal{PACMSE}
\acmVolume{3}
\acmNumber{ISSTA}
\acmArticle{ISSTA161}
\acmMonth{10}
\acmSubmissionID{issta26main-p1785-p}
\received{2026-01-29}
\received[accepted]{2026-04-16}

\begin{document}

\makeatletter\@namedef{T1/zi4/m/it}{<->ssub*zi4/m/n}\makeatother

\newcommand{\tool}{Mure\xspace}
\newcommand{\dms}{DM\#\xspace}

\crefformat{section}{\S#2#1#3}
\crefformat{subsection}{\S#2#1#3}
\crefformat{subsubsection}{\S#2#1#3}


\newcommand{\ie}{\textit{i.e.}\xspace}
\newcommand{\eg}{\textit{e.g.}\xspace}
\newcommand{\etc}{\textit{etc}\xspace}
\newcommand{\perse}{\textit{per se}\xspace}
\newcommand{\ala}{\textit{à la}\xspace}
\newcommand{\cf}{\textit{c.f.}\xspace}
\newcommand{\via}{\textit{via}\xspace}
\newcommand{\vs}{\textit{vs.}\xspace}
\newcommand{\etal}{\textit{et al.}\xspace}
\newcommand{\viceversa}{\textit{vice versa}\xspace}

\newcommand{\ali}[1]{\textcolor[rgb]{1.0,0.0,0.0}{#1}}
\newcommand{\ben}[1]{\textcolor[rgb]{0.0,0.0,1.0}{#1}}
\newcommand{\sasan}[1]{\textcolor[rgb]{0.0,1.0,1.0}{#1}}
\newcommand{\shibbir}[1]{\textcolor[rgb]{0.0,1.0,0.0}{#1}}

\SetKw{Continue}{continue}
\SetKw{Break}{break}
\SetKwInOut{Parameter}{parameter}
\SetKw{New}{new}
\SetKw{Null}{null}
\SetAlFnt{\scriptsize}
\SetAlCapFnt{\scriptsize}
\SetAlCapNameFnt{\scriptsize}
\setlength{\algomargin}{0.8em}



\definecolor{circleFillColor}{RGB}{153, 153, 255}
\newcommand{\encircled}[2][]{
  \tikz[baseline=(c.base)]
    \node[circle,
          draw=black,
          text=black,
          fill=circleFillColor,
          line width=0.3pt,
          inner sep=0.1ex,
          minimum size=1.0ex,
          font=\scriptsize,
          #1] (c) {#2};
}

\title{Provably Lossless Acceleration of DNN Mutation Testing \via Memoization}

\author{Ali Ghanbari}
\orcid{0000-0003-1471-2546}
\affiliation{%
  \institution{Auburn University}
  \department{Department of Computer Science and Software Engineering}
  \city{Auburn}
  \country{USA}
}
\email{ghanbari@auburn.edu}

\author{Ben Greenman}
\orcid{0000-0001-7078-9287}
\affiliation{%
  \institution{University of Utah}
  \department{Kahlert School of Computing}
  \city{Salt Lake City}
  \country{USA}
}
\email{blg@cs.utah.edu}

\author{Sasan Tavakkol}
\orcid{0000-0002-4111-4958}
\affiliation{%
  \institution{Google Research}
  \department{Google Research}
  \city{Irvine}
  \country{USA}
}
\email{tavakkol@google.com}

\author{Shibbir Ahmed}
\orcid{0000-0003-1183-883X}
\affiliation{%
  \institution{Texas State University}
  \department{Department of Computer Science}
  \city{San Marcos}
  \country{USA}
}
\email{shibbir@txstate.edu}

\begin{abstract}
Mutation analysis has recently reemerged in the context of deep neural networks (DNNs) as a promising, but notoriously costly, approach for assessing test dataset adequacy.
Existing techniques speed up DNN mutation testing through \textit{lossy} approximations that trade efficiency for mutation score accuracy.
This paper introduces \tool, the first provably lossless framework for accelerating DNN mutation testing \via memoization.
\tool is based on the idea that DNN mutants and the original model share substantial redundant computation, so during mutation testing, it executes only the mutated suffixes of each mutant and reuses the common prefix from the original model, which is computed only once.
We give a formal account of memoized mutation testing, and prove that \tool is sound, \ie, it produces results equivalent to exhaustive vanilla mutation testing, and identify basic conditions under which speed-up is guaranteed.
We have implemented \tool and evaluated it on 15 DNN models of various architectures, complexities, and sizes ranging from a few thousands to millions of parameters.
This provides empirical evidence that \tool reduces the computational cost of mutation testing by 44.54\%, on average.
We also observed that while state-of-the-art techniques tend to yield higher acceleration (up to 88.97\%, on average), they come at the cost of some error in mutation score.
We further analyze the effect of mutation generation selection ratio on the effectiveness of \tool and observed predictable reductions in memoization opportunities with increasing the percentage of mutated neurons.
We observed that \tool offers more than 20\% speed-up even when as high as 5\% of the neurons are mutated.
\end{abstract}



\maketitle

\section{Introduction}\label{sec:introduction}
\textit{Deep neural networks}~\cite{bib:lecun2015deep} (DNNs), have become important enabling technology in many modern software systems across many domains~\cite{bib:he2016deep,bib:hannun2014deep,bib:goldberg2022neural,bib:dahl2013large}.
Since DNNs are increasingly being used in safety- and mission-critical applications, from autonomous driving~\cite{grigorescu2020survey} to healthcare~\cite{esteva2019guide}, their quality assurance is critical.
Among various quality assurance approaches for DNNs~\cite{bib:liu2021algorithms,bib:huang2020survey}, testing is the most widely adopted one~\cite{bib:zhang2020machine}, where test data are either manually curated or automatically generated to satisfy specific adequacy requirements.
However, strong performance on a given test dataset does not necessarily guarantee robustness or generalization of the models and a principled means of evaluating the quality of test data is required.

Recently, \textit{mutation analysis}~\cite{bib:demillo1978hints,bib:hamlet1977testing} has been revisited in the context of DNNs~\cite{bib:shen2018munn,bib:ma2018deepmutation,bib:lu2022towards,bib:tambon2023probabilistic} as a promising technique for assessing test data quality (see~\cref{sec:back:mutation} for further background).
However, performing mutation analysis on DNNs is computationally expensive~\cite{bib:ma2018deepmutation,bib:hu2019deepmutation++,bib:shen2021boundary,bib:ghanbari2025using,bib:lyons2025on}.
The potential benefits of DNN mutation analysis has motivated researchers in recent years to investigate techniques for accelerating mutation analysis~\cite{bib:feng2022mutation,bib:li2021second,bib:li2022how,bib:ghanbari2025using,bib:ghanbari2024incite,bib:lyons2025on,bib:ghanbari2023mutation,bib:shen2021boundary}.
Unfortunately, all of the existing techniques, while effective, accelerate mutation analysis at the cost of incurring some error on the mutation score which can lead to errors in judging test dataset quality.
Such techniques are traditionally known as \textit{lossy}~\cite{bib:wang2017faster,bib:jia2010analysis}.
In conventional software systems, there exists a rich body of research on accelerating mutation testing~\cite{bib:pizzoleto2019systematic,bib:usaola2010mutation}, including both lossy and \textit{lossless} techniques, \ie, mutation analysis acceleration techniques that do not incur any mutation score error, which we also refer to as \textit{sound}.
However, since mutation testing for DNNs is relatively recent, lossless mutation analysis acceleration has not yet received much attention.

In this paper, we propose \tool, a provably lossless acceleration framework for DNN mutation testing.
The key insight behind \tool is that mutants and the original model share portions of their computation; rather than redundantly executing those shared portions for each mutant, \tool evaluates common computation once and \textit{safely} reuses it. 
Roughly speaking, given a DNN and its mutants, \tool first constructs a \textit{layer dependence graph} to capture how layers depend on each other.
Layer dependence graph is used to identify \textit{memoization boundaries} for the mutants, \ie, the collection of indices of the layers before the first mutated layer on which the layers on or after the first mutated layer depend.
\tool then groups mutants by their \textit{memoization boundaries}.
This way, we essentially group the mutants based on the deepest shared prefix of the mutants whose activations can be safely reused from the original model.
A \textit{memo table} is also constructed by instrumenting the original model once per each group guided by their shared memoization boundary.
Grouping mutants and constructing memo table once per group avoids creating large memo tables.
Mutants in each group are then \textit{chopped} right after their shared memoization boundary by replacing the shared prefix with a \textit{fabricated input layer}, which feeds inputs to the layers on or after the first mutated layer.
During mutation testing, the input for each chopped mutant is fetched from the memo table.
This way, only the mutated suffix of the mutants will be executed.
The outputs of the mutants are then used for calculating mutation score.

We formalize this process and prove that \tool is \textit{sound}, \ie lossless.
In other words, executing mutants through \tool is equivalent to executing the mutants exhaustively without memoization, aka, \textit{vanilla} mutation testing.
This implies that \tool produces mutation testing outcomes that are equivalent to vanilla approach, while avoiding redundant computations.
We prove a speed-up theorem showing that \tool is guaranteed to reduce execution cost whenever multiple mutants share a non-trivial prefix (see~\cref{sec:theory} for details).

We have implemented \tool in Python using Keras and Tensorflow frameworks.
It is applicable to supervised classifier Keras models with a wide range of architectures.
While our formal guarantees establish soundness and theoretical performance potential, they do not fully determine the practical effectiveness of \tool.
Specifically, some important factors are naturally empirical, \eg, the cost of model instrumentation, the overhead of memo table and mutant chopping, and external factors such as system load and stochasticity of mutation generation.
Moreover, to compare \tool to prior work, empirical comparison is necessary.
Thus, we conduct an extensive empirical evaluation using 15 DNN models with different architectures and complexities from different domains, including fully-connected neural networks (FCNNs), convolutional neural networks (CNNs) with and without residual blocks, and recurrent neural networks (RNNs).
Our experiments investigate two research questions.
First, we evaluate whether \tool delivers meaningful acceleration in practice and how it compares to state-of-the-art approaches.
We observed that, in our dataset, \tool reduces the cost of DNN mutation testing by 44.54\%, on average, all the while \textit{incurring no mutation score error}, as it was already established by our soundness theorem.
We further observed that other state-of-the-art techniques, \dms~\cite{bib:ghanbari2025using} and BSS~\cite{bib:shen2021boundary}, while provide higher acceleration (63.59\% and 88.97\%, on average, respectively), they incur mutation score error.
Second, we examine how mutation generation selection ratio with different values (\ie, mutating different percentage of neurons from the original model) affects memoization potential in \tool.
We observed a predictable decreasing trend: as the mutation generation selection ratio increases, the achievable speed-up gradually decreases due to reduced shared computation, but even in mutation generation selection ratios as high as 5\%, \tool still offers more than 20\% speed gain (see~\cref{sec:experiments} for more details).

In summary, this paper makes the following contributions.
\begin{itemize}
\item \textbf{Technique:} We propose \tool, the first memoization-based technique for accelerating DNN mutation testing. We have implemented \tool~\cite{bib:replica} to be applicable to a wide range of Keras models. \tool is publicly available in~\cite{bib:replica}.

\item \textbf{Theoretical Results:} We formally define the execution semantics of memoized mutation testing and prove that \tool is lossless and equivalent to vanilla mutation testing. We also establish a speed-up theorem that shows \tool is guaranteed to yield a strict reduction in computation cost whenever redundant prefix computation exists across mutants.

\item \textbf{Empirical Results:} Using 15 DNN models of varying sizes, architectures, and complexities, we evaluate the practicality of \tool in terms of efficiency and compare it to state-of-the-art mutation acceleration approaches.
Additionally, we study the impact of mutation generation selection ratio on the effectiveness of \tool.
\end{itemize}

\section{Background}\label{sec:background}
\subsection{Deep Neural Networks}\label{sec:back:dnn}
Deep neural networks (DNNs) use multiple interconnected layers of transformation functions to convert inputs to outputs, wherein each layer learns successively higher level of abstraction from the training data.
Each layer in a DNN consists of \textit{neurons}, which are essentially non-linear computational units that apply an \textit{activation function}, \eg, rectified linear unit (ReLU), to the weighted sum of their inputs.
The weighted edges connecting neurons are often called \textit{synapses}.

When every neuron in a layer connects to all neurons in the next layer, we refer to the network as a fully connected neural network (FCNN).
Besides FCNNs, other well-known DNN architectures include convolutional neural networks (CNNs) and recurrent neural networks (RNNs), both of which can also be represented as graphs of neurons and synapses.
For instance, a convolutional layer with a $8 \times 2 \times 2 \times 4$ filter can be interpreted as 4 neurons, each summing over 32 synapses.
Thus, throughout this paper we regard DNNs as graphs, allowing us to use the terms ``nodes'' and ``neurons,'' as well as ``edges'' and ``synapses,'' interchangeably.

\subsection{DNN Mutation Analysis}\label{sec:back:mutation}
Mutation analysis~\cite{bib:demillo1978hints,bib:hamlet1977testing} is a program analysis technique for evaluating test suite quality in conventional programs.
It relies on \textit{mutation operators}, aka \textit{mutators}, that systematically transform program elements to generate program variants called \textit{mutants}.
These mutants are tested using a test suite to compare their output to that of the original program.
If outputs differ, the mutant is marked \textit{killed}, otherwise it is regarded as \textit{survives}.
Some mutants survive because they are semantically equivalent to the original program, hence the name \textit{equivalent mutants}.
The effectiveness, \ie, mutation adequacy, of a test suite is measured by the \textit{mutation score}, which is defined as the ratio of killed to all non-equivalent mutants.
The higher the mutation score the stronger the test suite is.
Beyond test suite assessment, mutation analysis has been widely applied in various software engineering tasks~\cite{bib:papadakis2019mutation,bib:jia2010analysis}.

More recently, mutation analysis has been adapted for DNNs~\cite{bib:shen2018munn,bib:ma2018deepmutation}.
Like neuron coverage~\cite{bib:pei2017deepxplore} and its variants~\cite{bib:ma2018deepgauge}, and unlike surprise adequacy~\cite{bib:kim2019guiding}, DNN mutation analysis is a structural test dataset adequacy assessment method.
Like its conventional counterpart, DNN mutation analysis has had plenty of applications beyond its original use in test quality assessment.
These applications include adversarial sample detection and generation~\cite{bib:wang2019adversarial,bib:hu2023muten}, robustness evaluation~\cite{bib:hu2019deepmutation++}, test data prioritization~\cite{bib:wang2021prioritizing}, model accuracy estimation~\cite{bib:hu2023aries}, fault localization and repair~\cite{bib:ghanbari2023mutation,bib:sohn2023arachne,bib:wu2022genmunn} and modularity~\cite{bib:ghanbari2024incite}. 
Specialized mutation analysis frameworks further target particular DNN domains such as autonomous driving~\cite{bib:jahangirova2021quality} and reinforcement learning~\cite{bib:lu2022towards}.

DNN mutation analysis, whether \textit{source-based} (mutating the source code defining the DNN model) or \textit{model-based} (mutating the neurons or their learned weights directly), is extremely costly.
Given the potentials of DNN mutation analysis, several research projects in the past few years have aimed to reduce its costs (see~\cref{sec:related} for more details).
Following a taxonomy, common in the literature for mutation analysis acceleration of conventional programs~\cite{bib:wang2017faster,bib:jia2010analysis}, we can classify the DNN mutation analysis acceleration techniques into two general categories: (1) lossy techniques that accelerate mutation analysis at the cost of loss in mutation score; (2) lossless, aka sound, techniques that accelerate mutation analysis without impacting mutation score.
Virtually all the existing DNN mutation analysis acceleration techniques are lossy.
In this paper, we introduce a novel lossless mutation analysis technique based on memoization.

\section{\tool Approach}\label{sec:approach}
\begin{figure}
    \centering
    \includegraphics[width=\linewidth]{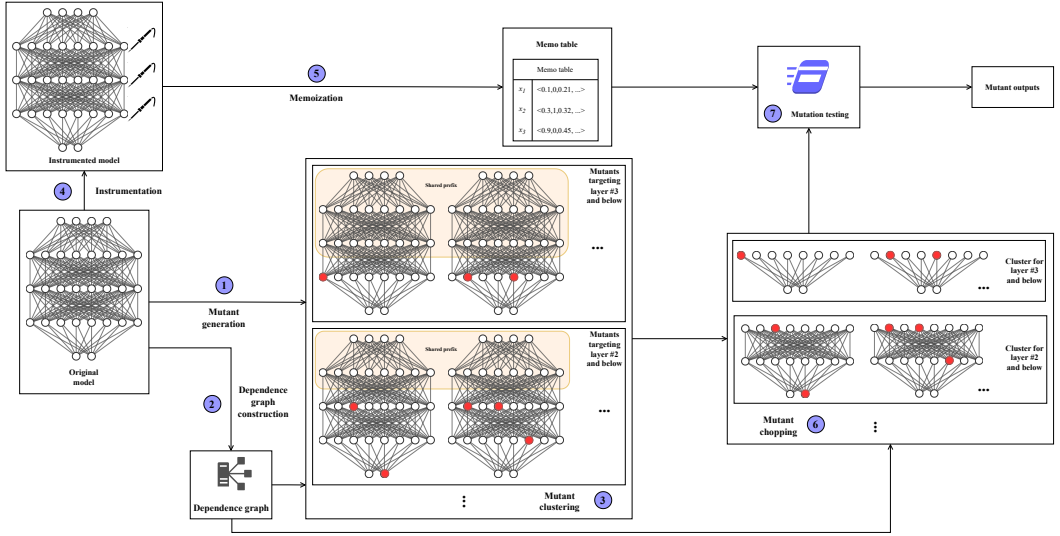}
    \caption{\tool workflow: artifacts are represented with boxes and arrows represent control flow} 
    \Description[\tool workflow]{\tool workflow: artifacts are represented with boxes and arrows represent control flow. Numbers are reference markers and do not correspond to the order of operations in \tool.}
    \label{fig:workflow}
\end{figure}
\tool is based on the idea that since multiple mutants as well as the original model share common portions, during testing we can evaluate those common portions once and reuse the results.
Specifically, if two mutants of a model involve mutating the layer(s) on or after layer indexed $n$, the layers $0,\dots, n-1$ are all unchanged are identical to the ones in the original model.
So, to test the two mutants, instead of testing them individually and reevaluating layers $0,\dots, n-1$ twice, we can evaluate the original model up to layer $n-1$ and then evaluate the mutants after layer $n-1$.
Later in~\cref{sec:theory} we will demonstrate that this is expected to result in a speed up without impacting the mutation score, and in~\cref{sec:experiments} we evaluate \tool empirically.

Fig.~\ref{fig:workflow} illustrates \tool's workflow at a high level.
Algorithm~\ref{alg:mure} is a more formal description.
\tool takes a model $M$, a test dataset $X$, and a mapping $\mathcal{M}$ of mutants of $M$.
The mapping $\mathcal{M}$ maps unique mutant identifiers to mutant objects.
Mutant generation, marked with \encircled{1} in Fig.~\ref{fig:workflow}, is explained in more details in~\cref{sec:app:mut-gen}.
For each mutant $M'$, the function $\mu$ returns the set of zero-based layer indices that have undergone mutation in $M'$.

\begin{figure}[t]
\begin{minipage}[t]{0.48\textwidth}
\begin{algorithm}[H]
\scriptsize
\SetKwBlock{Begin}{Method:}{end}
\caption{\tool memoized mutation testing}\label{alg:mure}
    \KwIn{Model $M$, test dataset $X$, a mapping $\mathcal{M}$ from mutant ids to mutants}
    \KwOut{A mapping from mutant ids to their output on $X$}
    \Begin{
        $dg\leftarrow \mathtt{getDepGraph}(M)$ \\
        $clusters \leftarrow \{\}$ \\
        \For{$(id, M') \in \mathcal{M}$} {
            $mmi\leftarrow \min(\mu(M'))$ \\
            $K \leftarrow \{j  \mid j \in dg[li] \wedge  j < mmi \wedge li \in mmi\dots|M.\mathtt{layers}|-1\}$ \\
            \uIf{$K\not\in \mbox{keys}(clusters)$} {
                $clusters[K] \leftarrow []$
            }
            $clusters[K].\mathtt{append}(id)$
        }
        $res \leftarrow \{\}$ \\
        \For{$(K, ids)\in clusters$} {
            $M_i\leftarrow \mathtt{instrument}(M, K)$ \\
            $MT\leftarrow M_i.\mathtt{predict}(X)$ \\
            \For{$id \in ids$} {
                $M' \leftarrow \mathcal{M}[id]$ \\
                $mmi\leftarrow \min(\mu(M'))$ \\
                $M'_\text{chopped}\leftarrow \texttt{chopMutant}(M', mmi, dg, K, MT)$ \\
                $res[id]\leftarrow M'_\text{chopped}.\mathtt{predict}([MT[k] \mid k\in K])$
            }
        }
        \Return {$res$}
    }
\end{algorithm}
\end{minipage}
\hfill
\begin{minipage}[t]{0.48\textwidth}
\begin{algorithm}[H]
\scriptsize
\SetKwBlock{Begin}{Method:}{end}
\caption{Dependence graph calculation: procedure \texttt{getDepGraph} used in Algorithm~\ref{alg:mure}}\label{alg:dep}
    \KwIn{Model $M$}
    \KwOut{A dictionary $dg$ mapping each layer index to a set of layer indices from which it receives inputs}
    \Begin{
            \uIf{$M$ is \texttt{Sequential}} {
                \Return{$\{i\mapsto \{i-1\}\mid 0< i < |M.\mathtt{layers}|\}\cup \{0\mapsto\emptyset\}$}
            }
            $dg\leftarrow\{\}$\\
                \For{$i$ in $0\dots|M.\mathtt{layers}|-1$} {
                    $L\leftarrow \emptyset$ \\
                    $l\leftarrow M.\mathtt{layers}[i]$ \\
                    \For{$n$ in $l$.\_inbound\_nodes} {
                        $L\leftarrow L\cup \{\mathtt{layerIndexOf}(n)\}$
                    }
                    $dg[i]\leftarrow L$
                }
            \Return{$dg$}
    }
\end{algorithm}
\end{minipage}
\Description[Two algorithms]{\tool memoized mutation testing and dependence graph calculation algorithms}
\end{figure}

\tool a constructs dependence graph for the model $M$ by calling \texttt{getDepGraph} (Line 2 in Algorithm~\ref{alg:mure} and step \encircled{2} in Fig.~\ref{fig:workflow}).
A dependence graph is a dictionary, mapping layer index $li$ to the set of layer indices $\{li_0,\dots,li_n\}$ if the layer at $li$ receives its inputs from layers indexed $li_0,\dots,li_n$, hence depends on them.
We present the details of this function in~\cref{sec:app:dep}.
\tool then clusters the mutants guided by the index of the earliest layer mutated in each mutant (step \encircled{3} in Fig.~\ref{fig:workflow}).
This process is formalized in Lines 3--9 in Algorithm~\ref{alg:mure}: for each mutant $M'$ with mutant id $id$, we find the layer indices that come before the earliest layer mutated by $M'$ such that some layer on or after the earliest mutated layer depends on them.
The activation values of these layers, indexed by set $K$, called \textit{memoization boundary}, will constitute the inputs to the \textit{chopped mutants}, described shortly.
Clusters are constructed by putting the mutant ids with equal $K$ in the same bucket.

After grouping the mutants based on their $K$ sets, \tool iterates through the clusters as follows.
It first instruments the original model $M$ at the layers indexed by $K$ (step \encircled{4} in Fig.~\ref{fig:workflow}).
The instrumented model is then applied on the given test dataset $X$ to construct a memo table (\encircled{5} in Fig.~\ref{fig:workflow}).
The memo table is essentially a dictionary mapping each data point in $X$ to the activation vectors of the layers indexed by $K$ at that data point.
Instrumentation is done by passing $M$ to the function \texttt{instrument}, which reads outputs of the layers indexed by $K$.
This process is explained in detail in~\cref{sec:app:instrument}.
After constructing the memo table for a group of mutants, \tool \textit{chops} the mutants within the cluster by discarding unmutated layers that come before the earliest mutated layer and replacing them with a single fabricated input layer (\encircled{6} in Fig.~\ref{fig:workflow}).
\tool does this while respecting dependencies between the layers.
The function \texttt{chopMutant}, explained in~\cref{sec:app:chopping}, is responsible for mutant chopping.
Each chopped mutant is then tested by retrieving their inputs from the memo table (\encircled{7} in Fig.~\ref{fig:workflow}).
The outputs of the chopped mutants are stored in a dictionary indexed by the mutant ids.
These steps are formalized in Lines 10--19 of Algorithm~\ref{alg:mure}.
In this paper, we use mutant outputs to calculate mutation score \ala Ma \etal~\cite{bib:ma2018deepmutation}, but the results could be used for virtually any other task that relies on DNN mutation, \eg, ~\cite{bib:ghanbari2023mutation,bib:ghanbari2024incite,bib:wang2019adversarial,bib:wang2021prioritizing}.

\subsection{Mutant Generation}\label{sec:app:mut-gen}
\tool takes a set of pre-generated mutants as input.
We use the mutation generation engine of DeepMutation ~\cite{bib:ma2018deepmutation} to generate the mutants, but any other mutation generator would also work.
Specifically, we use DeepMutation's \textit{Gaussian Fuzzing}, \textit{Neuron Activation Inverse}, \textit{Weight Shuffle}, and \textit{Neuron Effect Block} mutators to create a set of mutants for testing.
While DeepMutation offers more mutators, \eg, Layer Removal~\cite{bib:ma2018deepmutation}, previous studies~\cite{bib:jahangirova2020empirical,bib:wang2019adversarial} suggest that those are less effective in the sense that they are likely to create trivial mutants.
Thus, we use the aforementioned set of mutators only, some of which have also been used in other works~\cite{bib:wang2019adversarial,bib:ghanbari2024incite}.
These mutators may mutate more than one neuron at a time depending on the user configuration.
The degree of higher mutation impact in DeepMutation is controlled by a parameter called \textit{mutation generation selection ratio}, which specifies what percentage of neurons in a model should be selected for mutation.
The default value for this parameter is 1\%, \ie, 1\% of the neurons from the models are mutated.
We study the impact of mutation generation selection ratio on the effectiveness of \tool in~\cref{sec:experiments}.

\subsection{Constructing a Dependence Graph}\label{sec:app:dep}
Removing layers before the earliest mutation point is not always safe, especially in functional models where any layer can depend on any other layer and not necessarily the layer immediately before it.
\tool addresses this challenge by identifying dependencies between the layers so that if a layer on or after the earliest mutation point depends on layer(s) before the earliest mutation point, those layers are instrumented in the original model and their outputs are fed to the chopped models during memoized mutation testing.
Layer dependencies in \tool are represented as a \textit{dependence graph}.
A dependence graph is a mapping from layer indices to the set of layer indices on which the layer depends.
The procedure \texttt{getDepGraph}, which is presented in Algorithm~\ref{alg:dep}, is responsible for creating dependence graph.
This algorithm receives a Keras model as input and returns the dependence graph for the model.
In line 2, the algorithm checks if the model is a \texttt{Sequential} Keras model, in which case, the dependence graph simply maps each layer index to the layer index before it, and the 0\textsuperscript{th} layer is mapped to empty set, as the input layer does not depend on other layers (Line 3).
Otherwise, for models with more general architecture, \eg, a model with residual blocks, the algorithm constructs the dependence graph as follows.
It populates an empty mapping, by iterating through the layers and collecting the layer indices for the inbound nodes for each layer.
Keras provides an attribute \texttt{\_inbound\_nodes} for the layer objects that contains a list of nodes, \ie, neurons, whose outputs are fed into the given layer object.
In lines 8--9, the algorithm iterates through the inbound neurons and adds their layer indices to a set.
The function \texttt{layerIndexOf} returns the index of the layer containing a given neuron.
By adding the layer indices in a set, we remove duplicate layer indices in cases where a layer depends on multiple neurons from another layer.
After that, the algorithm maps the layer index to the set that it has just constructed.
The mapping is returned at Line 11.

\begin{figure}[t]
\centering
\begin{minipage}[t]{0.48\textwidth}
\begin{algorithm}[H]
\scriptsize
\SetKwBlock{Begin}{Method:}{end}
\caption{Model instrumentation: \texttt{instrument} procedure used in Algorithm~\ref{alg:mure}}\label{alg:instrument}
    \KwIn{Model $M$ and a list of layer indices $p$ serving as probe points}
    \KwOut{An instrumented models whose outputs are the activations at layers indexed by $p$}
    \Begin{
        \uIf{$M$ is \texttt{Sequential}} {
            $inp \leftarrow \mathtt{Input}(\mathtt{shape}=M.\mathtt{input\_shape})$ \\
            \uIf{$p = []$} {
                \Return {$\mathtt{Model}(\mathtt{inputs}=inp, \mathtt{outputs}=inp)$}
            }
            $x \leftarrow inp$ \\
            \For{ $i\in 0\dots |M.\mathtt{layers}|-1$ } {
                $x \leftarrow M.\mathtt{layers}[i](x)$ \\
                \uIf{$i \in p$} {
                     \Return {$\mathtt{Model}(\mathtt{inputs}=inp, \mathtt{outputs}=x)$}
                }
            }
            \textbf{error} ``Probe point not found''
        }
        $outs \leftarrow []$ \\
        \For{$i\in 0\dots |M.\mathtt{layers}|-1$} {
            $l\leftarrow M.\mathtt{layers}[i]$ \\
            \uIf{$i\in p$} {
                $outs.\texttt{append}(l.\texttt{output})$
            }
        }
        \Return {$\mathtt{Model}(\mathtt{inputs}=M.\mathtt{inputs}, \mathtt{outputs}=outs)$}
    }
\end{algorithm}
\end{minipage}
\hfill
\begin{minipage}[t]{0.48\textwidth}
\begin{algorithm}[H]
\scriptsize
\SetKwBlock{Begin}{Method:}{end}
\caption{Mutant chopping: procedure \texttt{chopMutant} used in Algorithm~\ref{alg:mure}}\label{alg:chopping}
    \KwIn{Mutant $M'$, minimum mutated layer index $mmi$, layer dependence graph $dg$, boundary layer indices $K$, and memo-table $MT$}
    \KwOut{$M'$ chopped at layer indexed $mmi$ taking inputs at that layer}
    \Begin{
        \uIf{$M'$ is \texttt{Sequential}} {
            $inp \leftarrow \mathtt{Input}(\mathtt{shape}=\mathtt{shape}(MT[k]))$, where $K=\{k\}$\\
            $x \leftarrow inp$ \\
            \For{$j \in mmi\dots|M.\mathtt{layers}|-1$} {
                $x \leftarrow M'.\mathtt{layer}[j](x)$
            }
            \Return{$\mathtt{Model}(\mathtt{inputs}=inp,\mathtt{outputs}=x)$}
        }
        $mutInp \leftarrow \{k\mapsto \mathtt{Input}(\mathtt{shape}=\mathtt{shape}(MT[k]))\mid k\in K\}$ \\
        $past\leftarrow \{\}$ \\
        \For{$j \in mmi\dots|M.\mathtt{layers}|-1$} {
            $l\leftarrow M'.\mathtt{layers}[j]$ \\
            $layerInp \leftarrow []$\\
            \For{$p \in dg[j]$} {
                \uIf{$p<mmi$} {
                    $layerInp.\mathtt{append}(mutInp[p])$
                } \uElse {
                    $layerInp.\mathtt{append}(past[p])$
                }
            } 
            $x \leftarrow l(\mathtt{inputs}=layerInp)$\\
            $past[j]\leftarrow x$
        } 
        \Return{$\mathtt{Model}(\mathtt{inputs}=[il\mid (\_, il)\in mutInp],\mathtt{outputs}=x)$}
    }
\end{algorithm}
\end{minipage}
\Description[Two algorithms]{Model instrumentation and mutant chopping algorithms}
\end{figure}

\subsection{Model Instrumentation and Memo Table Construction}\label{sec:app:instrument}
\tool records the output values, aka activations, of any layer before earliest mutation point that has dependent(s) on or after the earliest mutation point.
It does so by defining a new model, called \textit{instrumented model}, based on the existing \texttt{Sequential} or functional model.
The instrumented model receives the same inputs that the original model, but it outputs the activations of a selected list of layers that we call \textit{probe points}.
The procedure \texttt{instrument}, presented in Algorithm~\ref{alg:instrument}, creates an instrumented model given a Keras model $M$ and a list $p$ of probe points.
Algorithm~\ref{alg:mure} invokes this procedure by passing a memoization boundary $K$ as probe points.
If $M$ is \texttt{Sequential}, the algorithm creates an input layer that receives the same input shape as $M$.
If $p$ is empty, it means that the layer right after the input layer has been mutated, so the instrumented model acts as an identity function returning its inputs (Lines 4--5).
Otherwise, there is one probe point and the algorithm ``stitches'' layers together, \via layer invocation, to produce the output of the instrumented model (Lines 6--10).
Line 11 is a safeguard against passing an invalid $p$, \eg, the one that does not contain a valid layer index.
Handling functional models is easier in that there is no need for instantiation and stitching of the layers: the algorithm simply uses the collective outputs of the probe points as the output of the instrumented model (Lines 12--17).

The outputs of instrumented models are used to construct a memo table shared among all mutants with a common memoization boundary.
More concretely, given the test dataset $T$ and instrumented model $M_i$, a memo table is constructed by building the mapping $\{t\mapsto M_i(t)\mid t\in T\}$, where $M_i(t)$ returns the activations of $M_i$ at its probe points.

\subsection{Model Chopping}\label{sec:app:chopping}
Model chopping is the last preparation step before \tool executes memoized mutation testing over a cluster of mutants.
Given a mutant $M'$, the index $mmi$ of its earliest mutated layer, the dependence graph $dg$, the set of boundary layers $K$, and the memo table $MT$, the procedure \texttt{chopMutant}, described in Algorithm~\ref{alg:chopping}, constructs a new Keras model called the \textit{chopped mutant}.
Chopped mutant behaves like $M'$ on the portion of the network on and after layer index $mmi$, but it no longer contains any of the unmutated prefix layers.
Instead, it takes as input the activation values at the boundary layers in $K$, which are supplied at test time from the memo table (see~\cref{sec:app:instrument}).
Intuitively, \texttt{chopMutant} ``re-roots'' the mutant at the earliest mutation point while preserving all data dependencies from the removed prefix through explicit inputs.

For \texttt{Sequential} models, chopping is straightforward because the dependence structure is linear: the algorithm fabricates a single input corresponding to the boundary layer immediately preceding  $mmi$, and then reconnects the remaining layers in order, effectively treating the suffix starting at $mmi$ as a smaller sequential network whose first input is the memoized activation right before $mmi$ (Lines 2--7).
In functional models, the layers in the suffix may depend on multiple predecessors, some of which may lie in the prefix to be removed.
In this case, Algorithm~\ref{alg:chopping} creates one input layer for each boundary layer in $K$ (Line 8).
It then rebuilds the mutant from layer $mmi$ to the output layer by consulting the dependence graph as follows: if a layer's predecessor lies in the prefix, the algorithm uses the corresponding boundary input; otherwise, it uses the already reconstructed output of the predecessor (Lines 9--20).
This incremental reconstruction preserves the original dataflow of the mutant while discarding all computation before $mmi$.
The resulting chopped model behaves identically to the full mutant when evaluated with the memoized activations at the boundary layers, enabling \tool to avoid redundant recomputation of the shared prefix during mutation testing.

\section{Theoretical Results}\label{sec:theory}
In this section, we give a more formal account of the properties of the \tool memoized mutation testing algorithm.
Specifically, we present a proof of soundness, \ie, \tool does not change the mutation outcome (in other words, \tool is a lossless acceleration technique), and a proof that \tool is expected to result in speed-up.
We will first present the formalism for representing DNNs that both proofs rely on.
We would like to emphasize that both of these proofs rely on the fact that none of our mutators physically delete any of the layers or neurons in the model nor do they add new layers or neurons.
This is true for the mutators used in this paper (see~\cref{sec:app:mut-gen}).
We expect that this assumption will not diminish the usefulness of our theoretical results for most real-world situations because structure-altering mutators tend to generate trivial, easily-killable, mutants~\cite{bib:jahangirova2020empirical}, and are avoided by some works~\cite{bib:wang2019adversarial,bib:ghanbari2024incite}.

A DNN model is a directed acyclic graph $M=(V,E,\Phi,I,O)$, where $V=\{v_0, \dots, v_L\}$ is a set of nodes representing layers, $E\subseteq V\times V$ is a set of directed edges representing layer dependencies along which the tensors flow, $\Phi(v)$ is the layer function at $v$ mapping tensors from the previous layer(s) to an output tensor, $I\subseteq V$ is the non-empty set of input layers, and $O\subseteq V$ is the non-empty set of output layers.
In this definition, $I\neq O$.
In addition, there must be a topological ordering function $\tau: V\rightarrow \mathbb{N}$ such that $(v_i,v_j)\in E$ implies $\tau(v_i)<\tau(v_j)$, for all $v_i,v_j\in V$.

A \textit{mutant} $M'$, created by mutating $M$, is defined to be $M'=(V,E,\Phi',I,O)$, wherein $V$, $E$, $I$, and $O$ are unchanged, because none of the mutants alter the physical structure of the model.
The layer function, however, is defined as follows.
\begin{equation*}
    \Phi'(v)=\left\lbrace\begin{array}{ll}
        \Phi(v) &; v\in \mu(M')  \\
        \mbox{mutated layer function} &; \mbox{otherwise}
    \end{array}\right.
\end{equation*}
In this definition, $\mu(M')\subseteq V$ is a function that returns a subset of layers that have been the subject of mutation.
Based on this definition, $\Phi'$ is identical to $\Phi$ except at the mutated layers where $\Phi'$ use the layer function for the mutated layers.
The layer function for the mutated layers is left \textit{unspecified}, and none of our proofs need a precise characterization of the function.
The topological ordering $\tau$ allows us to formally define \textit{earliest mutated layer} of $M'$ as $mmi=\mbox{arg min}_{v\in \mu(M')}~\tau(v)$.

We define \textit{memoization boundary}, \ie, the required input set for a chopped mutant (defined shortly), as the set of all \textit{unmutated} layers whose outputs are needed to compute the mutated region.
This set is defined to be $K=\{p\mid (p,v)\in E~\mbox{such that}~\tau(v)\geq \tau(mmi)~\mbox{and}~\tau(p)<\tau(mmi)\}$.

The \textit{activation}, \ie, the output value, for each layer $v\in V$ in the original model $M$ is defined recursively as follows.
\begin{equation*}
    \mathbf{a}_v(\mathbf{x})=\left\lbrace\begin{array}{ll}
        \mathbf{x} & ; v\in I, \\
        \Phi(v)(\mathbf{a}_p(\mathbf{x})_{p\in Pred(v)}) & ; \mbox{otherwise}
    \end{array}\right.
\end{equation*}
where $Pred(v)=\{p\mid (p, v)\in E\}$ are the immediate predecessors of $v$.
We use the notation $M(\mathbf{x})$, defined as $(\mathbf{a}_v(\mathbf{x})_{v\in O})$, to represent the result of applying $M$ on a data point $\mathbf{x}$.
We define a memo table as function that maps each $k\in K$ to the tensor $\mathbf{a}_k(\mathbf{x})$, \ie, the activation of layer $k$ on input $\mathbf{x}$.
Similar to $\mathbf{a}_v$, we define $\mathbf{b}_v(\mathbf{x})$, the activation for each layer $v\in V$ in the mutated model $M'$ recursively using $\Phi'(v)$ in place of $\Phi(v)$ and $\mathbf{b}_p(\mathbf{x})$ instead of $\mathbf{a}_p(\mathbf{x})$.
We also define $M'(\mathbf{x})$ as $(\mathbf{b}_v(\mathbf{x})_{v\in O})$ representing the result of applying the mutant $M'$ on a data point $\mathbf{x}$.

A \textit{chopped mutant}, $M^\star$, is a DNN whose inputs consist of one tensor per boundary node $k\in K$, each having the same shape as the activation of layer $k\in K$ in the original model.
More concretely, a chopped mutant is a directed acyclic graph $M^\star=(V^\star,E^\star,\Phi^\star,I^\star,O)$, where $I^\star=\{u_k\mid k\in K\}$ where $u_k$ is a fresh input layer, $O$ is the set of output layer nodes as in the original model $M$, and $V^\star=\{v\mid v\in V~\mbox{and}~v~\mbox{is reachable from}~mmi~\mbox{in}~M\}\cup I^\star\cup O$.
The set $E^\star\subseteq V^\star\times V^\star$ is defined as follows.
For each $v\in V^\star$, $(p,v)\in E^\star$ if and only if (1) $(p,v)\in E$ and $p\in V^\star$, or (2) $p=u_k$ for some $k\in K$ and $(k,v)\in E$.
Finally, $\Phi^\star(v)$ is defined as follows.
\begin{equation*}
    \Phi^\star(v)=\left\lbrace\begin{array}{ll}
        \Phi'(v) &; v\in \mu(M'), \\
        \Phi(v) &; v\not\in \mu(M')~\mbox{and}~v\in (V^\star\cap V),\\
        \mathbf{id} &; v\in I^\star,
    \end{array}\right.
\end{equation*}
where $\mathbf{id}$ is the identity function returning the tensor it receives from the input without any computation.
The function $\Phi^\star$ behaves identical to $\Phi'$ for all mutated layers and it behaves identical to the original $\Phi$ for all unmutated layers.
The function models the behavior of input layers as identity functions.
Given this layer function, we define the \textit{activation} for each layer $v\in V^\star$ in the chopped model recursively as follows.
\begin{equation*}
    \mathbf{b}^\star_v(\mathbf{x})=\left\lbrace\begin{array}{ll}
        \mathbf{x} & ; v\in I^\star, \\
        \Phi^\star(v)(\mathbf{b}^\star_p(\mathbf{x})_{p\in Pred^\star(v)}) & ; \mbox{otherwise}
    \end{array}\right.
\end{equation*}
where $Pred^\star(v)=(Pred(v)\cap V^\star)\cup \{u_k\mid k\in K~\mbox{and}~k\in Pred(v)\}$. 
Finally, we define $M^\star(\mathbf{x})$ as $(\mathbf{b}^\star_v(\mathbf{x})_{v\in O})$ to represent the result of applying the chopped mutant $M^\star$ on a data point $\mathbf{x}$.

\subsection{Soundness}\label{sec:theory:soundness}
Our soundness proof relies on one more assumption that for any unmutated layer $v$, the function $\Phi(v)$ is deterministic and depends only on its inputs.
In particular, for inference in Keras, where \texttt{BatchNorm} and \texttt{Dropout} are in inference mode and no stateful RNNs are used or their states are reset, $\mathbf{a}_v(\mathbf{x})$ depends only on the inputs it receives, \textit{not} on prior test inputs or model state.
Please note that this assumption does not create any limitation for the applicability of our theoretical results, because as long as we run our models in inference mode, and do not load models with \texttt{stateful=True} (or set it to \texttt{False} before applying \tool), this assumption will hold.

The soundness proof states that the output of Algorithm~\ref{alg:chopping} is equal to that of the full mutant.
This way, instead of testing full mutants, one may test the chopped mutants and get the same mutation testing results.
Soundness directly follows from several lemmas described shortly: first we show that the output of Algorithm~\ref{alg:chopping} is equivalent to the theoretical model of chopped mutant that we defined in~\cref{sec:theory} (Lemma~\ref{lem:chopping:correct}), in that they both compute the same function, we then demonstrate that the theoretical model of chopped mutant is functionally equivalent to the full, un-chopped mutant (Lemma~\ref{lem:full-and-chopped:equiv}).
These results together imply that the output of Algorithm~\ref{alg:chopping} is functionally equivalent to the full, un-chopped mutant (Theorem~\ref{thm:soundness}), \ie, they both compute the same function.

Throughout this section, Let $M=(V,E,\Phi,I,O)$ be a DNN model, wherein layers set $V=\{v_0,\dots,v_L\}$ is ordered using Keras' topological order $\tau$, \ie, $\tau(v_i)=i$, for all $0\leq i\leq L$.
Before presenting our main results, we need to state and prove a helper lemma which concerns the correctness of Algorithm~\ref{alg:dep}.

\begin{lemma}[Dependence Graph Correctness]\label{lem:dg:correct}
    Let $dg=\mathtt{getDepGraph}(M)$.
    Then for every index $i$, we have $dg[i]=\{j\mid (v_j,v_i)\in E\}$.
\end{lemma}
\begin{proof}
    We proceed by a case analysis on the type of the model $M$: (1) $M$ is a \texttt{Sequential} model; (2) $M$ is a more general functional model.
    In case (1), the layers are linearly ordered according to $\tau$ and they are connected to each other in that order, so that all edges in $E$ have the form $(v_{i-1}, v_i)$.
    Additionally, according to Line 3 of Algorithm~\ref{alg:dep}, we have $dg[i]=\{i-1\}$, for $i>0$, and $dg[0]=\emptyset$.
    Thus, we can rewrite $dg[i]$ as $\{j\mid (v_j,v_i)\in E\}$, which is exactly the result that we were looking for.
    
    In case (2), while the layers are linearly ordered according to $\tau$, they are not necessarily connected to each other in that order.
    Based on lines 6--10 of Algorithm~\ref{alg:dep}, $dg[i]$ contains layer indices of the inbound neurons to the $i$-th layer, so $dg[i]=\{\tau(p)\mid (p,v_i)\in E\}$.
    Since $\tau(v_j)=j$, for all $0\leq j\leq L$, we can rewrite $dg[i]$ as $\{j\mid (v_j,v_i)\in E\}$.
\end{proof}

We now use this lemma to show that the output of Algorithm~\ref{alg:chopping} is functionally equivalent to the theoretical model of chopped mutant.

\begin{lemma}[Mutant Chopping Correctness]\label{lem:chopping:correct}
    Let $M'=(V,E,\Phi',I,O)$ be a mutant of $M$, and $mmi\in\{0,\dots,L\}$ be the index of the earliest mutated layer in $M'$.
    Let $K$ be the memoization boundary associated with $M'$ and $mmi$, $MT$ be the memo table associated with $M$ and $K$, and $M^\star=(V^\star,E^\star,\Phi^\star,I^\star,O^\star)$ be the chopped mutant, as defined in~\cref{sec:theory}.
    Let $dg=\mathtt{getDepGraph}(M)$, and $M^\star_{Alg}=\mathtt{chopMutant}(M',mmi,dg,K,MT)$, such that $M^\star_{Alg}=(V^\star_{Alg},E^\star_{Alg},\Phi^\star_{Alg},I^\star_{Alg},O^\star_{Alg})$.
    Then $M^\star$ and $M^\star_{Alg}$ compute the same function.
\end{lemma}
\begin{proof}
    We show that there exists a bijection $h:M^\star\rightarrow M^\star_{Alg}$, such that:
    \begin{enumerate}[label=(\alph*)]
        \item \textit{Nodes preserved}: For each boundary node $k\in K$, $h(u_k)$ is the input layer created in Algorithm~\ref{alg:chopping} for $k$. For each internal node $v\in V\cap V^\star$, $h(v)$ is the node constructed based on $v$.
        \item \textit{Edge structure preserved}: For each $p,v\in V$, $(p,v)\in E^\star\iff (h(p),h(v))\in E^\star_{Alg}$.
        \item \textit{Layer functions preserved}: For each $v\in V^\star$, $\Phi^\star(v)=\Phi^\star_{Alg}(h(v))$.
    \end{enumerate}
    The proof proceeds by a case analysis on the type of mutant $M'$: (1) the mutant $M'$ is \texttt{Sequential} model; (2) $M'$ is a more general functional model.
    \paragraph{\textbf{Case (1)}: (a) Nodes preserved}
    Since $M'$ is \texttt{Sequential}, we have $K=\{v_{mmi-1}\}$, and hence $I^\star=\{u_{v_{mmi-1}}\}$.
    Algorithm~\ref{alg:chopping} creates exactly one input node $inp=\mathtt{Input}(\mathtt{shape}=\mathtt{shape}(MT[v_{mmi-1}]))$ at Line 3.
    We define $h(u_{v_{mmi-1}})=inp$, which clearly is a bijective mapping.
    To prove that $h$ maps each node at the mutated suffix of the mutant to one and only one node in the chopped mutant, we first note that $V\cap V^\star=\{v_j\mid j\geq mmi\}$.
    At lines 4--6, Algorithm~\ref{alg:chopping} creates the chain: $x_{mmi-1}=inp$, $x_i=v_j(x_{i-1})$, for each $mmi\leq j\leq L$.
    Note that the notation $v_j(\cdot)$ is borrowed from Python, which is equivalent to calling \texttt{\_\_call\_\_} method of the layer object $v_j$.
    This call results in reusing layer $v_j$ in the created chopped mutant, without cloning the layer object, and wiring $x_{i-1}$ to $x_i$.
    We now define $h(v_j)=x_j$, for $j\geq mmi$.
    This is also a bijective mapping.
    \paragraph{(b) Edge structure preserved}
    For a \texttt{Sequential} mutant, $E^\star=\{(u_{v_{mmi-1}}, v_{mmi})\}\cup\{(v_{j-1},v_j)\mid j>mmi\}$.
    Algorithm~\ref{alg:chopping} constructs exactly these edges: At line 3 and the first iteration of the \textbf{for}-loop at line 5, it connects $inp=h(u_{v_{mmi-1}})$ to $h(v_{mmi})$.
    In subsequent iterations of the loop, it connects $h(v_j)$ to $h(v_{j+1})$.
    Therefore, $(p,v)\in E^\star\iff (h(p),h(v))\in E^\star_{Alg}$.
    \paragraph{(c) Layer functions preserved}
    For memoization boundary nodes, $\Phi^\star(u_{v_{mmi-1}})=\mathbf{id}$.
    The node $inp=\mathtt{Input}(\mathtt{shape}=\mathtt{shape}(MT[v_{mmi-1}]))$ created by Algorithm~\ref{alg:chopping} also implements identity.
    Thus, $\Phi^\star_{Alg}(inp)=\mathbf{id}$.
    For each internal node $v_j$, $\Phi^\star(v_j)=\Phi'(v_j)$.
    Algorithm~\ref{alg:chopping} also reuses the nodes from the mutant $M'$, through the calling and chaining mechanism described above.
    So, $\Phi^\star_{Alg}(h(v_j))=\Phi'(v_j)$, for each $j\geq mmi$.
    Therefore, $\Phi^\star(v)=\Phi^\star_{Alg}(h(v))$, for each $v\in V^\star$.
    \paragraph{\textbf{Case (2)}: (a) Nodes preserved}
    For chopped mutant $I^\star=\{u_k\mid k\in K\}$.
    At line 8, Algorithm~\ref{alg:chopping} constructs the mapping $mutInp[k]=\mathtt{Input}(\mathtt{shape}=\mathtt{shape}(MT[k]))$, for each $k\in K$.
    We set $h(u_k)=mutInp[k]$, for $k\in K$.
    These are clearly bijective mappings.
    To prove that $h$ maps each node at the mutated suffix of the mutant to one and only one node in the chopped mutant, we first note that $V\cap V^\star=\{v_j\mid j\geq mmi\}$.
    At lines 10--19, Algorithm~\ref{alg:chopping} defines $past[j]=v_j(\mbox{inputs constructed at lines 12--17})$ for each $j\in \{mmi,\dots,L\}$.
    As we discussed in the \texttt{Sequential} case, $v_j(\cdot)$, where $v_j\in V$, reuses $v_j$ in the construction of $M^\star_{Alg}$ without cloning the layer.
    We define $h(v_j)=past[j]$, which clearly is a bijective mapping.
    \paragraph{(b) Edge structure preserved}
    We note that $Pred^\star(v)=(Pred(v)\cap V^\star)\cup \{u_k\mid k\in K\cap Pred(v)\}$.
    It follows from Lemma~\ref{lem:dg:correct} that $Pred(v_j)=\{v_p\mid p\in dg[j]\}$.
    At lines 14--18, for each $p\in dg[j]$, where $j\in\{mmi,\dots, L\}$, Algorithm~\ref{alg:chopping} connects $mutInp[p]=h(u_{v_p})$ to $past[j]=h(v_j)$, if $p<mmi$ (hence $v_p\in K$), and it connects $past[p]=h(v_p)$ to $past[j]=h(v_j)$, if $p\geq mmi$.
    After the \textbf{for}-loop of line 13, the incoming edges for $past[j]$ are emanating from the nodes $\{h(v_p)\mid p\in dg[j]~\mbox{and}~p\geq mmi\}\cup\{h(u_{v_k})\mid k\in dg[j]~\mbox{and}~k<mmi\}$.
    Therefore, we have $(p,v_j)\in E^\star\iff (h(p),h(v_j))\in E^\star_{Alg}$
    \paragraph{(c) Layer functions preserved}
    This proof proceeds by a straightforward generalization of the sequential case to compare boundary nodes $u_k$ and the nodes $mutInp[k]$.
    
    
\end{proof}

We next show that the theoretical chopped mutant is equivalent to the full, un-chopped mutant.

\begin{lemma}[Full and Chopped Mutant Equivalence]\label{lem:full-and-chopped:equiv}
    Let $M'=(V,E,\Phi',I,O)$ be a mutant of $M$, and $mmi\in\{0,\dots,L\}$ be the index of the earliest mutated layer in $M'$.
    Let $K$ be the memoization boundary associated with $M'$ and $mmi$, $MT$ be the memo table associated with $M$ and $K$, and $M^\star=(V^\star,E^\star,\Phi^\star,I^\star,O^\star)$ be the chopped mutant, as defined in~\cref{sec:theory}.
    We have $M^\star(\mathbf{a}_k(\mathbf{x})_{k\in K})=M'(\mathbf{x})$, for every data point $\mathbf{x}$.
\end{lemma}
\begin{proof}
    By expanding the definitions of model applications, we can see that the proof goal is $(\mathbf{b}^\star_v(\mathbf{a}_k(\mathbf{x})_{k\in K})_{v\in O})=(\mathbf{b}_v(\mathbf{x})_{v\in O})$, for every data point $\mathbf{x}$.
    We prove this via a stronger proposition that considers all layers (including the output layer): for every $t\in \{mmi, mmi+1, \dots, L\}$, and every data point $\mathbf{x}$, we have $\mathbf{b}^\star_{v_t}(\mathbf{a}_k(\mathbf{x})_{k\in K})=\mathbf{b}_{v_t}(\mathbf{x})$.
    Note that $\{v_t\mid mmi\leq t\leq L\}=V^\star\cap V$.
    We proceed by induction on the topological ordering $\tau$ over $V^\star\cap V$.

    \paragraph{Induction Basis}
    Let $t=mmi$.
    In this case, we want to prove $\mathbf{b}^\star_{v_{mmi}}(\mathbf{a}_k(\mathbf{x})_{k\in K})=\mathbf{b}_{v_{mmi}}(\mathbf{x})$.
    By expanding the $\mathbf{b}^\star_{v_{mmi}}$ and $\mathbf{b}_{v_{mmi}}$, we get $\Phi^\star(v_{mmi})(b^\star_p(\mathbf{a}_k(\mathbf{x})_{k\in K})_{p\in Pred^\star(v_{mmi})})=\Phi'(v_{mmi})$ $(\mathbf{b}_p(\mathbf{x})_{p\in Pred(v_{mmi})})$.
    Since $v_{mmi}$ is, by definition, mutated, \ie, $v_{mmi}\in \mu(M')$, $\Phi^\star(v_{mmi})$ evaluates to $\Phi'(v_{mmi})$.
    Additionally, since $Pred^\star(v_{mmi})\subseteq I^\star$, \ie, all the predecessors of the earliest mutated layer are the inputs to the chopped mutant, by definition, $b^\star_p$ is the identity function for each $p\in Pred^\star(v_{mmi})$.
    Furthermore, since nodes in $Pred(v_{mmi})$ are not mutated, we have $\mathbf{b}_p(\mathbf{x})_{p\in Pred(v_{mmi})}=\mathbf{a}_p(\mathbf{x})_{p\in Pred(v_{mmi})}$.
    This is because the full mutant is equal to the original model before the earliest point of mutation.
    Thus, we get equal expressions on the two sides of equation: $\Phi'(v_{mmi})(\mathbf{a}_p(\mathbf{x})_{p\in Pred(v_{mmi})})=\Phi'(v_{mmi})(\mathbf{a}_p(\mathbf{x})_{p\in Pred(v_{mmi})})$.

    \paragraph{Induction Step}
    Take an index $t>1$, where $mmi\leq t< L$.
    Assume that for every index $j<t$ the above proposition holds.
    That is to say, $\mathbf{b}^\star_{v_j}(\mathbf{a}_k(\mathbf{x})_{k\in K})=\mathbf{b}_{v_j}(\mathbf{x})$, for every data point $\mathbf{x}$.
    This is our \textit{induction hypothesis}.
    
    We need to prove $\mathbf{b}^\star_{v_t}(\mathbf{a}_k(\mathbf{x})_{k\in K})=\mathbf{b}_{v_t}(\mathbf{x})$ for every data point $\mathbf{x}$.
    We start by expanding the definitions of $\mathbf{b}^\star_{v_t}$ and $\mathbf{b}_{v_t}$.
    Given that $v_t$ is not an input layer, we have $\Phi^\star(v_t)(\mathbf{b}^\star_p(\mathbf{a}_k(\mathbf{x})_{k\in K})_{p\in Pred^\star(v_t)})=\Phi'(v_t)(\mathbf{b}_p(\mathbf{x})_{p\in Pred(v_t)})$.
    Now we have two cases to examine: (1) $\tau(p)\geq mmi$, meaning that the predecessor of $v_t$ is not an input to the chopped mutant; (2) $\tau(p)< mmi$, meaning that the predecessor of $v_t$ is outside of the chopped region, so $p$ must be an input to the chopped mutant (the left-hand side of the equation), \ie, we have $p\in I^\star$ for the chopped mutant and $p\in K$ for the full mutant (the right-hand side of the equation).
    
    Case (1) directly follows from the induction hypothesis.
    To prove case (2), we fix $p\in K$ and consider $u_p\in I^\star$ as the corresponding input layer for the chopped mutant.
    So, what we want to prove is $\Phi^\star(v_t)(\mathbf{b}^\star_{u_p}(\mathbf{a}_p(\mathbf{x})))=\Phi'(v_t)(\mathbf{b}_p(\mathbf{x}))$.
    Note that $\mathbf{a}_p$ denotes the memo table entry that is fed to $u_p$.
    Now since $\mathbf{b}^\star_{u_p}=\mathbf{id}$, our proof obligation reduces to $\Phi^\star(v_t)(\mathbf{a}_p(\mathbf{x}))=\Phi'(v_t)(\mathbf{b}_p(\mathbf{x}))$.
    But since $p\in K$, which is outside of the mutated region, we have $\mathbf{b}_p(\mathbf{x})=\mathbf{a}_p(\mathbf{x})$.
    This is because the full mutant is equal to the original model before the earliest point of mutation.
    Thus, we get $\Phi^\star(v_t)(\mathbf{a}_p(\mathbf{x}))=\Phi'(v_t)(\mathbf{a}_p(\mathbf{x}))$.
    To complete the proof, we do yet another case analysis: (a) $v_t\in\mu(M')$, meaning that $v_t$ is a mutated layer; (b) $v_t\not\in\mu(M')$, meaning that $v_t$ is not mutated.
    In the first case, we have $\Phi^\star(v_t)=\Phi'(v_t)$, by definition.
    In the second case, we know that $v\in(V^\star\cap V)$, so both $\Phi^\star(v_t)$ and $\Phi'(v_t)$ will evaluate to $\Phi(v_t)$.
    This completes the proof.
\end{proof}

Finally, we are in a position where we can state our main soundness result.

\begin{theorem}[Soundness]\label{thm:soundness}
    Let $M'=(V,E,\Phi',I,O)$ be a mutant of $M$, and $mmi\in\{0,\dots,L\}$ be the index of the earliest mutated layer in $M'$.
    Let $K$ be the memoization boundary associated with $M'$ and $mmi$, $MT$ be the memo table associated with $M$ and $K$, as defined in~\cref{sec:theory}.
    Let $dg=\mathtt{getDepGraph}(M)$, and $M^\star_{Alg}=\mathtt{chopMutant}(M',mmi,dg,K,MT)$, such that $M^\star_{Alg}=(V^\star_{Alg},E^\star_{Alg},\Phi^\star_{Alg},I^\star_{Alg},O^\star_{Alg})$.
    For every data point $\mathbf{x}$, we have $M^\star_{Alg}((\mathbf{a}_k(\mathbf{x})_{k\in K})=M'(\mathbf{x})$.
\end{theorem}
\begin{proof}
    Directly follows from Lemmas~\ref{lem:chopping:correct} and~\ref{lem:full-and-chopped:equiv}.
\end{proof}

\subsection{Acceleration}\label{sec:theory:accel}
We now prove that testing chopped mutants is faster than testing full mutants.
We use an abstract evaluation-cost function $cost(F)$, that represents the computational cost of evaluating DNN model $F$ once on a given input data point $\mathbf{x}$ and that satisfies the following \textit{cost function axioms}:
\begin{itemize}
    \item \textbf{C1: Non-negativity:} $cost(F)\geq 0$ for every DNN model $F$.
    \item \textbf{C2: Additivity:} If a model $F$ can be decomposed as a prefix $P$ and a suffix $S$, such that running $F$ on an input amounts to running $P$ on the input and then feeding its output(s) to $S$, then $cost(F)=cost(P)+cost(S)$.
    \item \textbf{C3: Invariance:} If two DNN models $F$ and $G$ have isomorphic computation graphs, \ie, they have the same layers, same connectivity, and same parameters, then $cost(F)=cost(G)$.
\end{itemize}

As presented in~\cref{sec:approach}, \tool clusters mutants based on their $K$ values, which is the set of memoization boundary nodes.
We first express the acceleration results for only one, non-empty, cluster, denoted $\mathcal{C}$ in this section, as a lemma.
Proof of acceleration of multiple mutant clusters directly follows from this result.
Let $M\downarrow K$ denote the model obtained by running $M$ only up to the boundary nodes $K$, which is the shared prefix needed to calculate the memo table for the mutants in $\mathcal{C}$, \ie, $(\mathbf{a}_k(\mathbf{x})_{k\in K})$.

\begin{lemma}[In-Cluster Acceleration]\label{lem:accel}
    Given a DNN model $M$, a set $\mathcal{C}$ of mutants, a set $K$ of memoization boundary nodes, we have $cost(M\downarrow K)+\sum_{M'\in \mathcal{C}}cost(M^\star_{Alg})\leq \sum_{M'\in \mathcal{C}} cost(M')$, for every input data point $\mathbf{x}$.
    This inequality is strict if $|\mathcal{C}|>1$ and $cost(M\downarrow K)>0$.
\end{lemma}
\begin{proof}
    Take an arbitrary input data point $\mathbf{x}$.
    Since all the mutants in $\mathcal{C}$ share the same earliest mutated layer index, \ie, $mmi$, the following holds for the mutants in $\mathcal{C}$ by construction: for all $M'\in \mathcal{C}$, all layers $v$ with $\tau(v)<mmi$ are unmutated, and hence are identical to $M$.
    Let $P$ denote the prefix submodel consisting of these unmutated layers up to the boundary nodes $K$.
    By definition, evaluating $M\downarrow K$ on $\mathbf{x}$ is equivalent to evaluating $P$ on $\mathbf{x}$.
    Thus, we have:
    \begin{equation}\label{eq:accel:1}
        cost(M\downarrow K)=cost(P).
    \end{equation}
    For each mutant $M'\in \mathcal{C}$, we denote the suffix of submodel, \ie, everything after the memoization boundary nodes, as $S_{M'}$.
    Then each full mutant $M'$ decomposes into $P$ and $S_{M'}$.
    To produce the output of $M'$ we first apply the prefix $P$ on the input data point $\mathbf{x}$, and its activations at $K$ are then fed to $S_{M'}$.
    Then by C2, we have: $cost(M')=cost(P)+cost(S_{M'})$, for all $M'\in \mathcal{C}$.
    But $S_{M'}$ is the same as $M^\star$, and according to Lemma~\ref{lem:chopping:correct}, $M^\star$ is isomorphic to $M^\star_{Alg}$, so by C3, we have $cost(S_{M'})=cost(M^\star)=cost(M^\star_{Alg})$.
    Thus, we have: $cost(M')=cost(P)+\sum_{M'\in \mathcal{C}}cost(M^\star_{Alg})$.
    The total cost of \textit{vanilla mutation testing}, \ie, testing with no memoization, is then $\sum_{M'\in \mathcal{C}}cost(M')$.
    Using~\ref{eq:accel:1}, we can rewrite this as: $\sum_{M'\in \mathcal{C}}cost(M')=\sum_{M'\in \mathcal{C}}(cost(P)+cost(M^\star_{Alg}))=|\mathcal{C}|\times cost(P)+\sum_{M'\in \mathcal{C}}cost(M^\star_{Alg})$
    Unlike vanilla mutation testing, \tool evaluates $M\downarrow K$ once at the cost of $cost(P)$, and for each mutant $M'\in \mathcal{C}$ it runs the chopped mutant $M^\star_{Alg}$.
    Thus, the total cost of memoized mutation testing for the cluster $\mathcal{C}$ is $cost(P)+\sum_{M'\in \mathcal{C}}cost(M^\star_{Alg})$.
    By subtracting the memoized mutation-testing cost from the vanilla mutation-testing cost, we obtain $\sum_{M' \in \mathcal{C}} cost(M') - (cost(P) + \sum_{M' \in \mathcal{C}} cost(M^\star_{\mathit{Alg}}))$.
    This simplifies to $|\mathcal{C}|\times cost(P)+\sum_{M'\in \mathcal{C}}cost(M^\star_{Alg}) - (cost(P)+\sum_{M'\in \mathcal{C}}cost(M^\star_{Alg}))$, and with a bit of algebra, we get $(|\mathcal{C}| - 1)\times cost(P)$.
    By C1, $cost(P) \ge 0$, and therefore the difference is non-negative. Hence, $cost(P) + \sum_{M' \in \mathcal{C}} cost(M^\star_{\mathit{Alg}}) \le \sum_{M' \in \mathcal{C}} cost(M').$ Finally, by Equation~(1), $cost(P)=cost(M\downarrow K)$, so $cost(M\downarrow K) + \sum_{M' \in \mathcal{C}} cost(M^\star_{\mathit{Alg}}) \le \sum_{M' \in \mathcal{C}} cost(M')$. Moreover, if \(|\mathcal{C}|>1\) and \(cost(M\downarrow K)>0\), then $(|\mathcal{C}|-1)\cdot cost(P) = (|\mathcal{C}|-1)\cdot cost(M\downarrow K)>0$. Thus, the memoized cost is strictly smaller than the vanilla cost, and the inequality above is strict.
\end{proof}

The above lemma concerns a single cluster of mutants that share the same $mmi$ and $K$.
We can generalize this result to multiple clusters in a straightforward manner.

\begin{theorem}[Acceleration]\label{thm:accel:gen}
    Given a DNN model $M$ and a set of mutants $\{M^{(1)}, \dots, M^{(n)}\}$.
    Assume that these mutants are clustered into $p$ clusters $\mathcal{C}_1, \dots, \mathcal{C}_p$, which their respective boundary node sets $K_1, \dots, K_p$ and earliest mutated layer indices ${mmi}_1, \dots {mmi}_p$. We have $\sum_{i\in\{1,\dots,p\}} cost(M\downarrow K_i)+\sum_{M'\in \mathcal{C}_i}cost(M^\star_{{Alg}_i})\leq \sum_{i\in\{1,\dots,n\}} cost(M^{(i)})$, for every input data point $\mathbf{x}$.
    This inequality is strict, if at least one $|\mathcal{C}_i|>1$ and $cost(M\downarrow K_i)>0$.
\end{theorem}
\begin{proof}
    Directly follows from axiom C1 and Lemma~\ref{lem:accel}.
\end{proof}

Theorem~\ref{thm:accel:gen} proves that there is a speed up given a set of \textit{pre-constructed} chopped mutants, but it does not determine how much acceleration should be expected.
In particular, this theorem does not take into account the costs associated with (1) constructing chopped mutants; (2) forming clusters; (3) constructing memo table; and (4) Keras modeling building as well as running the models on a physical computer.
Additionally, it does not tell us how the percentage of neurons mutated impacts the amount of acceleration.
Therefore, conducting experiments is necessary for getting a sense of how much speed up should be expected given all the aforementioned factors.

\section{Experiments}\label{sec:experiments}
Our theoretical results in~\cref{sec:theory} provide guarantees that \tool is lossless, and that \tool is expected to result in speed up.
However, many details such as the amount of speed up as well as how it compares against related work are not clear.
We answer the remaining questions empirically by investigating the following research questions (RQs).
\begin{itemize}
    \item \textbf{RQ1:} How much speed up does \tool offer compared to the baseline approaches?
    \item \textbf{RQ2:} How does mutation generation selection ratio impact speed gain by \tool?
\end{itemize}

In RQ1, we compare acceleration achieved by \tool to that of two state-of-the-art approaches \dms~\cite{bib:ghanbari2025using} and BSS~\cite{bib:shen2021boundary} on 7 models.
We observed that \tool consistently accelerates DNN mutation testing across all models, with an speed gain of 44.54\%.
While \dms and BSS provide higher acceleration (63.59\% and 88.97\%, on average, respectively), they incur mutation score error.

In RQ2, we provide empirical evidence that the speed-up provided by \tool decreases gradually as the mutation generation selection ratio increases.
The results also indicate that while mutation generation selection ratio reduces memoization potential, it does not eliminate the practicality of \tool even in mutation generation selection ratios as high as 5\%.

\subsection{Baseline Approaches}\label{sec:exp:baselines}
We compare \tool to three baseline approaches: (1) vanilla approach, which tests mutants exhaustively; (2) \dms~\cite{bib:ghanbari2025using}, which clusters mutants based on behavioral similarity to test representatives from each cluster and reuse their results for all the mutants in the clusters they represent; and (3) BSS~\cite{bib:shen2021boundary}, which tests the mutants using only a subset of data points around the decision boundary of the original model rather than the whole test dataset.
These techniques are the most effective representatives of a broad range of related techniques.
\dms is the latest representative of techniques that speed up mutation testing by testing fewer mutants, \eg, neuron and mutant clustering~\cite{bib:lyons2025on,bib:ghanbari2024incite}, random mutant selection~\cite{bib:ghanbari2023mutation}, mutator selection~\cite{bib:feng2022mutation}, and higher-order mutation~\cite{bib:li2022how}.
Furthermore, BSS represents techniques that accelerate mutation testing by reducing mutant inputs, \eg, random sample selection~\cite{bib:ghanbari2025using}.

\subsection{Measures}\label{sec:exp:measures}
We quantize \textit{speed-up}, aka \textit{acceleration} or \textit{gain}, as $\frac{t_v-t}{t_v}$, where $t_v$ is the vanilla mutation testing time, while $t$ is the mutation testing time for the acceleration techniques, \eg, \tool, \dms.
Both $t_v$ and $t$ are measured in seconds.
We would like to emphasize that the measured time for \tool includes all preprocessing and bookkeeping overheads, including dependence-graph construction, mutant clustering, model instrumentation, memo-table construction, and mutant chopping.
Therefore, the reported gains are net gains after accounting for any overhead \tool may have.

The second measure that we use is \textit{mutation score}.
We calculate mutation score \ala Ma \etal~\cite{bib:ma2018deepmutation}: $\frac{1}{|M|\times|L|}\sum_{\mu\in M}|\mbox{killingLabels}(\mu)|$, where $M$ is the set of mutants that is obtained by applying a set of mutators on a given DNN model $m$, and $L=\{\mbox{label(t)}\mid t\in T\}$ is the set of labels in a test dataset $T$, wherein the function `$\mbox{label}$' returns the ground-truth label for the data points in $T$.
For any mutant $\mu\in M$, $\mbox{killingLabels}(\mu)$ is defined to be $\{\mbox{label}(t)\mid t\in T~\mbox{and}~\mbox{kill}(\mu, t)\}$.
A mutant $\mu$ is said to be \textit{killed} by a test data point $t\in T$, denoted by the predicate $\mbox{kill}(\mu,t)$, if $\mbox{argmax}(m(t))=\mbox{label}(t)$ and $\mbox{argmax}(\mu(t))\neq\mbox{label}(t)$, where $\mbox{argmax}(m(t))$, or $\mbox{argmax}(\mu(t))$, denotes the label predicted by the model $m$, or its mutant $\mu$, for $t$.

For lossy approaches, the \textit{mutation score error} or \textit{loss}, \ie, the percentage of deviation of accelerated mutation score from vanilla mutation score, is calculated as $\frac{|MS_V - MS|}{MS_V}$, where $MS_V$ is the mutation score obtained using vanilla approach, while $MS$ represents the mutation score obtained using a lossy acceleration technique.
For \tool, mutation score error is zero by design, and it is \textit{not} because DeepMutation-style mutation score is aggregate and potentially obscures minor behavioral deviations.
It is solely because chopped mutants generated by \tool produce exactly the same outcome as original mutants.

\subsection{Benchmark and Setup}\label{sec:exp:benchmark}
Table~\ref{tab:benchmark} lists the DNN models used in the experiments, where each row represents a model.
We have used six types of DNN architectures in our experiments: 4-layer FCNN, LetNet-5~\cite{bib:lecun1998gradient}, ResNet-10~\cite{bib:he2016deep}, MobileNetV2~\cite{bib:sandler2018mobilnetv2}, and RNN with LSTM layers.
We have use standard LetNet-5 architecture that represents the family of CNN models with no residual blocks, such as AlexNet~\cite{bib:krizhevsky2012imagenet} and VGGNet~\cite{bib:simonyan2015very}.
We have also implemented ResNet-10 and MobileNetV2 based on their respective standard architectures, representing models with residual block of different complexity and layouts.
Our RNN architecture consists of an embedding layer, two LSTM layers, and an output layer with softmax activation.

\begin{table} 
    \centering
    \caption{Models benchmark. \# Train and \# Test represent the number of data points in the train and test datasets, respectively, and \# Cls denotes the number of classes/labels in the test dataset.}\label{tab:benchmark}
    \resizebox{0.6\textwidth}{!}{
        \begin{tabular}{|c|c||r|r|r|r|r|r|}
            \hline
            \multicolumn{1}{|c|}{\textbf{Architecture}} & \textbf{Dataset} & \multicolumn{1}{c|}{\textbf{Scope}} & \multicolumn{1}{c|}{\textbf{Size}} & \multicolumn{1}{c|}{\textbf{\# Train}} & \multicolumn{1}{c|}{\textbf{\# Test}} & \multicolumn{1}{c|}{\textbf{\# Cls}} & \multicolumn{1}{c|}{\textbf{Test Acc.}} \\
            \hline\hline
            \multicolumn{1}{|c|}{\textbf{LeNet-5}} & \textbf{SVHN} & \multicolumn{1}{r|}{\multirow{7}{*}{RQ1}} & 62006 & 73257 & 26032 & 10    & 0.8440 \\
            \cline{1-2}\cline{4-8}    \multicolumn{1}{|c|}{\multirow{2}{*}{\textbf{MobileNetV2}}} & \textbf{Caltech-101} &       & 334886 & 7316  & 1828  & 102   & 0.6461 \\
            \cline{2-2}\cline{4-8}          & \textbf{CIFAR-10} &       & 287690 & 50000 & 10000 & 10    & 0.7233 \\
            \cline{1-2}\cline{4-8}    \multicolumn{1}{|c|}{\multirow{2}{*}{\textbf{ResNet10}}} & \textbf{Caltech-101} &       & 5033446 & 7316  & 1828  & 102   & 0.6723 \\
            \cline{2-2}\cline{4-8}          & \textbf{CIFAR-10} &       & 287690 & 50000 & 10000 & 10    & 0.8132 \\
            \cline{1-2}\cline{4-8}    \multicolumn{1}{|c|}{\multirow{2}{*}{\textbf{RNN}}} & \textbf{IMDB} &       & 2642562 & 2642562 & 25000 & 2     & 0.8083 \\
            \cline{2-2}\cline{4-8}          & \textbf{Reuters} &       & 2648282 & 8982  & 2246  & 90    & 0.6399 \\
            \hline
            \multicolumn{1}{|c|}{\multirow{4}{*}{\textbf{FCNN}}} & \textbf{EMNIST} & \multicolumn{1}{r|}{\multirow{8}{*}{RQ2}} & 45676 & 124800 & 20800 & 26    & 0.8794 \\
            \cline{2-2}\cline{4-8}          & \textbf{FMNIST} &       & 44860 & 60000 & 10000 & 10    & 0.8779 \\
            \cline{2-2}\cline{4-8}          & \textbf{KMNIST} &       & 44860 & 60000 & 10000 & 10    & 0.8698 \\
            \cline{2-2}\cline{4-8}          & \textbf{MNIST} &       & 44860 & 60000 & 10000 & 10    & 0.9742 \\
            \cline{1-2}\cline{4-8}    \multicolumn{1}{|c|}{\multirow{4}{*}{\textbf{LeNet-5}}} & \textbf{EMNIST} &       & 45786 & 124800 & 20800 & 26    & 0.9197 \\
            \cline{2-2}\cline{4-8}          & \textbf{FMNIST} &       & 44426 & 60000 & 10000 & 10    & 0.8853 \\
            \cline{2-2}\cline{4-8}          & \textbf{KMNIST} &       & 44426 & 60000 & 10000 & 10    & 0.9410 \\
            \cline{2-2}\cline{4-8}          & \textbf{MNIST} &       & 44426 & 60000 & 10000 & 10    & 0.9899 \\
            \hline
        \end{tabular}
    }
\end{table}

We have trained these model architecture on various datasets of different complexities, such as MNIST~\cite{bib:deng2012mnist}, a dataset of hand-written digits with classes 0-9, Fashion MNIST~\cite{bib:xiao2017fashion} (FMNIST), a dataset with ten fashion classes, the digit section of Kuzushiji MNIST~\cite{bib:clanuwat2018deep} (KMNIST), a dataset of hand-written Japanese digits 0-9, and SVHN~\cite{bib:netzer2011reading}, a dataset of real-world images for 10-class classification of digits. 
We trained ResNet-10 and MobileNetV2 model architectures on more complex datasets such as CIFAR-10~\cite{bib:krizhevsky2009learning}, consisting of $32\times32$ color images in 10 different classes, and Caltech-101~\cite{bib:lazebnik2006beyond}, consisting $224\times224$ color images belonging to 102 different object categories.
Lastly, we used Reuters~\cite{bib:lewis1997reuters} and IMDB~\cite{bib:maas2011learning} datasets for training the RNN models.
Reuters is a dataset for 90-class classifiers for documents with news articles, and IMDB is a dataset for binary sentiment classification for movie reviews.

As indicated by the column ``Scope'' in the table, we are using more complex models for comparing \tool to state-of-the-art approaches, while simpler FCNN and LeNet-5 models are used in RQ2, making it feasible to run \tool and vanilla mutation analysis thousands of times

In the rest of the paper, we use the combination of model name and training dataset identifier to uniquely identify each of the 12 models, \eg, FCNN-FMNIST, denotes the model with FCNN architecture trained on FMNIST dataset.

We have used two identical Dell Precision workstations with AMD Ryzen Threadripper @ 2.7 GHz CPU, 1 TB of RAM to conduct our experiments.
Both machines run Ubuntu 22.04.4 LTS and no GPUs are used in our experiments.

\subsection{Answering RQ1}\label{sec:exp:rq1}
In this RQ, we measure the speed-up offered by \tool and compare it to that of two other approaches, namely \dms and BSS.
We run these techniques on 7 DNN models of varying sizes and complexities.
The results are reported in Table~\ref{tab:rq1-results}.
Since \dms and BSS are lossy approaches, in the sense that they accelerate mutation analysis at the cost of some error in the calculated mutation score, for these approaches we have also reported the mutation score error in the table.
Since training, as well as mutation generation, involve randomness, we have repeated our measurement 300 times (=10 rounds of training $\times$ 10 rounds of mutation generation $\times$ 3 rounds of mutation testing) and the averaged values are reported in Table~\ref{tab:rq1-results}.

\begin{table} 
    \centering
    \caption{\tool \vs state-of-the-art approaches in terms of speed gain over vanilla approach. Mutation score loss is reported for lossy approaches.}\label{tab:rq1-results}
    \resizebox{0.5\textwidth}{!}{
        \begin{tabular}{|c|c||r|r|r|r|r|}
            \hline
            \multicolumn{1}{|c|}{\multirow{2}{*}{\textbf{Architecture}}} & \multirow{2}{*}{\textbf{Dataset}} & \multicolumn{1}{c|}{\textbf{\tool}} & \multicolumn{2}{c|}{\textbf{\dms}} & \multicolumn{2}{c|}{\textbf{BSS}} \\
        \cline{3-7}          &       & \multicolumn{1}{c|}{\textbf{Gain}} & \multicolumn{1}{c|}{\textbf{Gain}} & \multicolumn{1}{c|}{\textbf{Loss}} & \multicolumn{1}{c|}{\textbf{Gain}} & \multicolumn{1}{c|}{\textbf{Loss}} \\
            \hline\hline
            \multicolumn{1}{|c|}{\textbf{LeNet-5}} & \textbf{SVHN} & 58.96\% & 60.90\% & 2.75\% & 80.86\% & 1.77\% \\
            \hline
            \multicolumn{1}{|c|}{\multirow{2}{*}{\textbf{MobileNetV2}}} & \textbf{Caltech-101} & 33.28\% & 50.65\% & 3.73\% & 97.27\% & 10.02\% \\
        \cline{2-7}          & \textbf{CIFAR-10} & 65.36\% & 85.39\% & 2.06\% & 96.44\% & 3.58\% \\
            \hline
            \multicolumn{1}{|c|}{\multirow{2}{*}{\textbf{ResNet10}}} & \textbf{Caltech-101} & 44.42\% & 79.05\% & 6.07\% & 94.41\% & 4.30\% \\
        \cline{2-7}          & \textbf{CIFAR-10} & 75.57\% & 88.91\% & 4.04\% & 95.20\% & 3.69\% \\
            \hline
            \multicolumn{1}{|c|}{\multirow{2}{*}{\textbf{RNN}}} & \textbf{IMDB} & 17.90\% & 56.56\% & 2.31\% & 94.14\% & 0.12\% \\
        \cline{2-7}          & \textbf{Reuters} & 16.29\% & 23.67\% & 2.09\% & 64.46\% & 11.46\% \\
            \hline
        \end{tabular}
    }
\end{table}%

Across the studied architectures, Table~\ref{tab:rq1-results} shows that while \dms and BSS achieve higher gain than \tool on average (63.59\% and 88.97\%, respectively, compared to 44.54\% of \tool), they do so by incurring some inaccuracy in mutation score.
Specifically, we can observe that in this dataset, \dms results in an average of 3.29\% mutation score error and BSS incurs a bit more loss of 4.99\%, on average.
Although the average acceleration by \tool is relatively lower than the other approaches, it does not incur any mutation score error by construction.

These contrasting characteristics suggest different use cases.
Lossy techniques such as \dms and BSS are appropriate when mutation testing is used as a fast heuristic signal, \eg, in exploratory or highly iterative workflows where approximate adequacy estimates are sufficient.
\tool is more appropriate when preserving the exact mutation testing outcomes is important, \eg, in mutation-guided test prioritization, debugging, or repair.
Although the average mutation-score losses of \dms and BSS are moderate in our experiments, they are not uniformly negligible: the maximum observed losses reach 12.88\% and 19.13\%, respectively.
Such deviations may affect downstream analyses and applications.
\tool avoids this trade-off by preserving the vanilla mutation testing result exactly while still cutting mutation testing time nearly in half.

To assess the stability of these results, we additionally computed per-model 95\% confidence intervals and non-parametric statistical tests across repeated runs.
The confidence intervals are narrow across models, which indicates that the observed speedups are stable under stochastic training and mutation generation.
Pairwise Mann-Whitney $U$-tests~\cite{bib:mann1947test} with Holm correction~\cite{bib:holm1979simple} show that the differences among approaches are statistically significant for all models.
The full confidence intervals table is included in our replication package~\cite{bib:replica}.

\subsection{Answering RQ2}\label{sec:exp:rq2}
To answer RQ2, we measure the speed-up offered by \tool for each round of mutation generation with a certain \textit{mutation generation selection ratio}.
As we discussed in~\cref{sec:app:mut-gen}, the amount of mutation in DeepMutation is controlled by its mutation generation selection ratio parameter.
The default value for this parameter is 1\%, but here we measure \tool speed-up against 3\% and 5\% mutation generation selection ratio.
We observed that values greater than 5\% generated so many low-quality mutants that were filtered out during DeepMutation's mutant quality assessment, thereby causing the program to loop for days without generating any mutants.

To account for the inherent randomness in both model training and mutation generation, we trained each model independently 10 times.
For each trained model, we generated mutants 10 times for each 1\%, 3\%, and 5\% mutation generation selection ratio values.
Finally, \tool is executed 3 times for each mutant set.
This resulted in 900 independent measurements per model configuration, substantially reducing the risk that our findings are artifacts of stochastic variation in training convergence, random weight perturbations, \etc.
These data are plotted in Fig.~\ref{fig:rq2-plots}, wherein the speed-up values are averaged for all 900 runs.

We can see from the figure that there is an inverse relationship between the mutation generation selection ratio and the acceleration achieved by \tool.
When the percentage of mutated neurons is small (\eg, 1\%), \tool consistently yields substantial gains, reducing testing cost by approximately 40\%–55\% for FCNN models and 60\%–65\% for LeNet-5 models on average.
As the percentage increases, these gains diminish monotonically.
This trend is expected: mutation generation selection ratio reduces the amount of shared computation between mutants and the original model, thereby shrinking the reuse opportunity that \tool exploits.
In other words, as more neurons are mutated, mutations are more likely to affect earlier layers. This shortens the network prefix that can be memoized and, consequently, limits the achievable speed-up.

We also examined dispersion across runs for each mutation generation selection ratio.
The boxplots and means-and-95\%-CI plots in the replication package~\cite{bib:replica} show the same monotonic trend as Fig.~\ref{fig:rq2-plots}: speedup decreases as the mutation generation selection ratio increases, but the trend is stable rather than being driven by outliers.
Additionally, mutation score error remains zero across \textit{all} ratios and models.
We would like to emphasize that running \tool on GPU may make the mutation score calculated by \tool to drift from vanilla due to floating-point errors in most GPUs.

These results provide empirical evidence that mutation generation selection ratio does not limit the practicality of \tool.
Instead, its impact is gradual and predictable: while increasing mutation ratio reduces memoization potential, \tool continues to offer considerable acceleration across all studied models for any practically useful mutation generation selection ratio.
\begin{figure}[t]
    \centering
    \includegraphics[width=\textwidth]{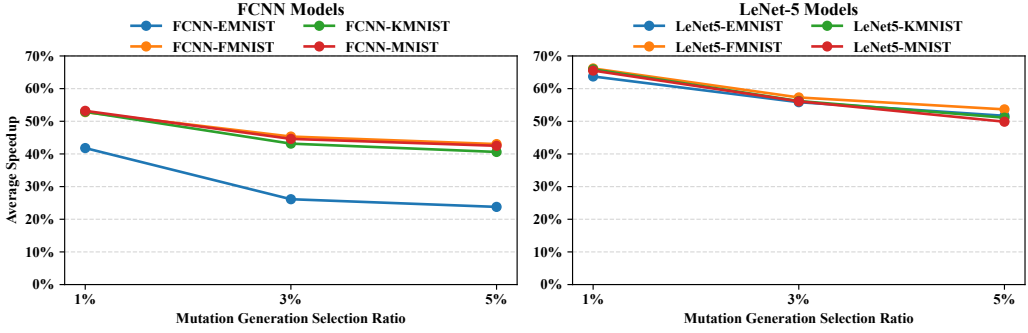}
    \caption{Average speed-up offered by \tool \vs mutation generation selection ratio}
    \Description[Average speed-up vs. mutation generation selection ratio]{Average speed-up vs. mutation generation selection ratio}
    \label{fig:rq2-plots}
\end{figure}

\section{Discussion}\label{sec:discussion}
\subsection{Practical Implication}\label{sec:discussion:implication}
We would like to emphasize that \tool, by itself, does not define a new testing adequacy criteria nor is it expected to expose failures that vanilla mutation testing would miss.
Rather, \tool preserves the exact semantics of vanilla mutation testing while reducing its cost.
This distinction matters because DNN mutation analysis is often used not for quantizing test adequacy, but as a source of information for downstream tasks such as test prioritization, adversarial sample generation, robustness evaluation, fault localization and repair, modularity analysis, and data augmentation.
In such settings, replacing vanilla mutation testing with a lossy approximation may change the downstream behavior of the technique.
\tool, instead, allows such workflows to use the \textit{same} mutation-testing outcome as vanilla mutation testing, but \textit{faster}.

\subsection{Threats to Validity}\label{sec:discussion:threats}
While we have provided formal proof for soundness and speed-up for \tool, most factors, such as randomness in higher order mutation and the overhead of mutant chopping and related bookkeeping tasks are difficult to model mathematically.
Therefore, we resort to empirical evaluation to take such factors into account, and like most empirical evaluations, our evaluation is subject to certain threats to validity.
Additionally, since the proofs are not machine-checked, they may contain error.
To mitigate the risks regarding potential errors in our proofs, we have checked our formalization and proofs multiple times, and have included them as part of the main text of the paper.

The first threat to validity of our empirical evaluation comes from the fact that DNN training and mutation generation are stochastic processes, \ie, mutation score calculated for one round of training and mutation generation can be different from that of another round.
To account for such randomness, we have trained each model 10 times and each model undergoes 10 rounds of mutation generation, resulting 100 sets of mutants for each model architecture-dataset pair.
Additionally, to account for randomness associated with system load, which we expect to be more predictable than randomness in training and/or mutation generation, we have executed \tool and other baseline approaches 3 times.
Finally, we have included all the models, measurements, and tools in our replication package~\cite{bib:replica} so that other researchers can reproduce our results.

Another threat is concerned with the range of mutation generation selection ratios we have studied in RQ2.
We observed that increasing mutation generation selection ratio beyond 5\% disturbs models so deeply that it turns them to trivial mutants, and we were not able to produce mutants even after running DeepMutation's mutation generation engine for days.
It could be possible that further experimentation above the 5\% ratio yields models that, given enough time, can generate non-trivial mutants to test \tool's behavior.
It may also be fruitful to rethink the mutation strategy in DeepMutation to improve success at higher selection ratios.

\section{Related Work}\label{sec:related}
Mutation analysis of DNNs is quite costly, \eg, mutating even the simplest DNN for classification of MNIST images results in hundreds of mutants, each of which must be tested against thousands of test data points.
While some speedup can be obtained using GPU support for numerical calculation, as provided by libraries such as Keras and NumPy, mutation remains an expensive process~\cite{bib:hu2019deepmutation++,bib:jahangirova2020empirical,bib:ghanbari2023mutation,bib:ghanbari2024incite}.
Motivated by the many potential practical applications of DNN mutation analysis, researchers have proposed several methods for reducing its costs.
Some of these methods are almost directly ported from mutation analysis of conventional programs into DNNs, while others take advantage of the unique characteristics of DNNs.
For example, Feng \etal~\cite{bib:feng2022mutation} and Wang \etal~\cite{bib:wang2023fine} identify sufficient subsets of existing mutators from the literature~\cite{bib:ma2018deepmutation,bib:shen2018munn} to avoid redundant mutants.
Ghanbari \etal~\cite{bib:ghanbari2023mutation} adopts random mutant selection, and Li \etal~\cite{bib:li2022how,bib:li2021second} reduces the cost of mutation testing with higher-order mutants.
Meanwhile, the technique introduced by Shen \etal~\cite{bib:shen2021boundary} relies on the assumption that mutants of a DNN model are more likely to produce different results for the test data points around the decision boundary of the model, so it samples the test data that lie at the decision boundary of the model under test.
Ghanbari~\cite{bib:ghanbari2024incite} proposed to accelerate mutation analysis by generating fewer mutants through clustering the neurons.
This approach together with mutant clustering based on the similarity of mutated weights have also been studied recently~\cite{bib:lyons2025on}.
More recently, Ghanbari and Tavakkol~\cite{bib:ghanbari2025using} proposed the use of Fourier analysis to generate comparable signatures of mutant behavior to use for mutant clustering and selecting one representative of each mutant for testing.

\section{Conclusions and Future Work}\label{sec:confuture}
This paper introduces \tool, the first provably lossless memoization-based acceleration framework for mutation testing of DNNs.
We formalized the execution semantics of memoized mutation testing and established theoretical guarantees that \tool preserves the exact mutation score of vanilla mutation testing, regardless of network architecture, dataset, or mutation operator.
We also provide a formal cost analysis and prove that \tool guarantees acceleration under certain conditions on mutant overlap.
Our empirical evaluation across a diverse set of DNN architectures show that \tool achieves substantial performance gains by reducing mutation testing cost by 44.54\%, on average, without incurring any error on mutation score.
Moreover, our analysis of mutation generation selection ratio with different values provide empirical evidence that while increasing mutation generation selection ratio gradually reduces memoization potential, \tool continues to deliver meaningful and predictable acceleration.

This paper opens several promising avenues for future exploration beyond the opportunities pointed out in the paper.
We plan to extend \tool to additional classes of architectures, including large transformer-based systems, to investigate whether new forms of structural sharing may further amplify memoization benefits.
We also we aim to explore adaptive memoization strategies, dynamic analysis techniques to identify additional reuse opportunities, and hybrid integration with lossy acceleration approaches to provide tunable points in the speed–accuracy design space.

\section*{Data Availability}
All our measurements and artifacts (such as software, extra tables, and figures) are available in our replication package~\cite{bib:replica}.

\section*{Acknowledgments}
We would like to thank Anonymous ISSTA 2026 Reviewers for their valuable feedback.
The first author is partially supported by the NSF grant \#2446393.
Any opinions, findings, and conclusions or recommendations expressed in this material are those of the authors and do not necessarily reflect the views of the NSF or their employers.

\bibliographystyle{ACM-Reference-Format}
\bibliography{main}

\end{document}